\documentclass[final,5p,times,twocolumn]{elsarticle}

\usepackage{amsmath,amsthm,amssymb,bm}
\usepackage{subfigure}
\usepackage{subfigmat}
\usepackage{color}
\usepackage{framed}
\usepackage[normalem]{ulem}
\usepackage{setspace}
\usepackage{enumitem}
\usepackage{array,tabularx,longtable}
\usepackage{ulem}

\allowdisplaybreaks

\newcommand{\norm}[1]{\left\lVert#1\right\rVert}
\newcommand{\abs}[1]{\left|#1\right|}

\newcommand{\ubar}[1]{\text{\b{$#1$}}}

\newcommand{\sgn}[1]{{\rm sgn}({#1})}
\newcommand{\numberthis}{\addtocounter{equation}{1}\tag{\theequation}}

\newtheorem{asmp}{Assumption}
\newtheorem{thm}{Theorem}

\newtheorem{lem}{Lemma}

\begin{document}

\begin{frontmatter}

\title{Adaptive Control with Rate-Limited Integral Action \\
for Systems with Matched, Time-Varying Uncertainties}

\author{Ying-Chun Chen$^a$, Craig Woolsey$^a$} 

\affiliation{organization={Kevin T. Crofton Department of Aerospace \& Ocean Engineering},
            addressline={Virginia Tech},
            city={Blacksburg},
            postcode={24060},
            state={Virginia},
            country={USA}}

\begin{abstract}
This paper considers the problem of controlling a piecewise continuously differentiable system subject to time-varying uncertainties. The uncertainties are decomposed into a time-invariant, linearly-parameterized portion and a time-varying unstructured portion. The former is addressed using conventional model reference adaptive control. The latter is handled using disturbance observer-based control. The objective is to ensure good performance through observer-based disturbance rejection when possible, while preserving the robustness guarantees of adaptive control. A key feature of the observer-based disturbance compensation is a magnitude and rate limit on the integral action that prevents fast fluctuations in the control command due to the observer dynamics.
\end{abstract}


\end{frontmatter}

\section{Introduction}
Adaptive control \cite{Krstic.AdaptiveControl,Hovakimyan.L1AdaptiveControl,Hovakimyan&Cao.L1AdaptiveControl,Narendra&Annaswamy.StableAdaptiveSystems,Lavretsky&Wise.RobustAdaptiveControl} is a robust control method that can effectively compensate for structured uncertainties with constant unknown parameters. It can also preserve closed-loop stability in the presence of time-varying disturbances as long as the disturbances are bounded. While stability is guaranteed in many such cases, however, control performance can degrade because of those time-varying disturbance terms. Variants of adaptive control have been proposed to deal with different types of unknown time-varying terms, as in \cite{Pagilla.ASMEDynSys2004,Chen.TAC2021,Marino.Automatica2003}, but the aim is to tolerate the disturbances rather than to recover the performance of a control law designed for a nominal, undisturbed system.

An intuitive approach for recovering the desired control performance is to directly compensate for disturbances using disturbance observer-based control (DOBC) \cite{Li&Chen.DOBC}, a control scheme that cancels disturbances using disturbance estimates. An advantage of such a strategy is that the disturbance estimator is designed independently of the controller and appended to that controller afterwards, which may be more expedient than simultaneously constructing an estimator and a controller. A typical DOBC design assumes the system to be controlled is totally known but subject to an external disturbance that is independent of the system state, estimating that external disturbance and counteracting it by the estimates; for examples, see \cite{Li&Chen.DOBC} and \cite{Back2008.DOBC,Khalil&Praly.HGOBC}. Advanced DOBC design further allows parametric uncertainties to appear in the system dynamics,resulting in the methods that consist of both the active rejection of external disturbances and suppression of system uncertainties; see \cite{Chen.DOSMC,Tao.DOSMC,Ginoya.DOSMC,Chen.JDSMC20,Chen.AJC20} for example. This paper adopts the same concept, extending conventional model reference adaptive control (MRAC) \cite{Krstic.AdaptiveControl,Hovakimyan.L1AdaptiveControl} to incorporate DOBC, a concept we will call \emph{disturbance observer-based adaptive control} (DOBAC). Rather than specify a particular disturbance observer structure, the paper describes a class of disturbance observers that may be incorporated into a DOBAC scheme.

Many disturbance observers assume the disturbance dynamics are benign and can therefore be ignored \cite{Li&Chen.DOBC}, though if the structure for the disturbance dynamics is known, the requirement for slow variation may be relaxed \cite{Chen2004.DOBC,Kim2010.DO}. Without knowledge of the dynamic structure, one may need to adopt a high observer gain \cite{Khalil.HGO,Chen&Woolsey.TCST23} to ensure fast-varying disturbances can be properly estimated, if any are present. The need to consider fast disturbance dynamics is not just academic. In many cases, the disturbance to be estimated is a ``lumped disturbance" that depends on not only external disturbances but also on the system state and input, increasing the possibility of fast variations due to the control system's dynamics. Classical observer-based control, for example, features fast-varying control signals due to observer dynamics \cite{ChenAndWoolsey.ACC22}. As a type of observer-based control, DOBC can suffer from poor disturbance estimation because of fast disturbance observer dynamics. To address that challenge, this paper proposes magnitude- and rate-limited integral action in the disturbance estimation scheme.

Section~\ref{sec:ControlDesign} describes the system structure and the control framework, which is the conventional MRAC structure augmented with a disturbance rejection term. The disturbance observer is introduced in Section~\ref{sec:DisturbanceRejection}, followed by a discussion in Section~\ref{sec:DifficultyInDisturbanceEstimation} about a challenge that can arise in estimating the disturbance. In Section~\ref{sec:IntuitiveDisturbanceRejection}, based on the concept of DOBC, we propose the design of the disturbance rejection term. In Section~\ref{sec:PerformanceAnalysis}, we analyze the overall performance of the resulting DOBAC. The idea is demonstrated in Section~\ref{sec:Example} using a mass-spring-damper system. Section~\ref{sec:KEtaRemarks} provides some remarks about a critical design parameter. Conclusions are presented in Section~\ref{sec:Conclusion}.

\section{System description and the control framework} \label{sec:ControlDesign}
Consider the following system whose state $\bm{x}(t) \in \mathbb{R}^n$ evolves under the influence of a control signal $u(t) \in \mathbb{R}$ and subject to an unknown, bounded disturbance $d(t) \in \mathbb{R}$:
\begin{equation} \label{eq:NonlinearSystem_Disturbance}
\begin{split}
    \dot{\bm{x}} &= \bm{A}\bm{x} + \bm{b}\Lambda(\bm{V}^T\bm{\phi}_V(\bm{x})+\bm{W}^T\bm{\phi}_W(\bm{x})+u+d)
\end{split}
\end{equation}
The constant terms $\bm{A}\in\mathbb{R}^{n \times n}$, $\bm{V}\in\mathbb{R}^{m_V}$, and $\bm{W}\in\mathbb{R}^{m_W}$ are unknown, while the constant vector $\bm{b}\in\mathbb{R}^n$ is known. The term $\Lambda\in\mathbb{R}$ is also an unknown constant, but $\sgn{\Lambda}$ is known. The known nonlinearities $\bm{\phi}_V:\mathbb{R}^{n}\to\mathbb{R}^{m_V}$ and $\bm{\phi}_W:\mathbb{R}^{n}\to\mathbb{R}^{m_W}$ are piecewise continuously differentiable.

Consider the following
\begin{align} \label{eq:ReferenceModel}
    \text{\emph{Reference model:}} \quad
    \begin{cases}
        \dot{\bm{x}}_{\rm r} = \bm{A}_{\rm r}\bm{x}_{\rm r} + \bm{b}{\Lambda_{\rm r}}r \\
        \bm{e} := \bm{x}-\bm{x}_{\rm r}
    \end{cases}
\end{align}
The Hurwitz matrix $\bm{A}_{\rm r}\in\mathbb{R}^{n \times n}$ and the input scaling $\Lambda_{\rm r} \in \mathbb{R}$ are chosen by the control designer, along with the piecewise continuously differentiable reference input $r(t) \in \mathbb{R}$, to generate a desired trajectory $\bm{x}_{\rm r}(t)$ for the system~\eqref{eq:NonlinearSystem_Disturbance} to track. The control objective is to have
\begin{align} \label{f:trackingError_Performance}
    \norm{\bm{e}} = \norm{\bm{x}-\bm{x}_{\rm r}} \leq \epsilon_{\rm r} \quad\forall\; t \geq T_{\rm r}
\end{align}
where the tracking error tolerance $\epsilon_{\rm r}$ and the convergence time $T_{\rm r}$ reflect the desired reference tracking performance.

The system~\eqref{eq:NonlinearSystem_Disturbance} may be rewritten as
\begin{equation} \label{eq:NonlinearSystem_Disturbance_Rewritten}
\begin{split}
    \dot{\bm{x}} &= \bm{A}_{\rm r}\bm{x} + \bm{b}\Lambda_{\rm r}(\bm{V}_{\rm r}^T\bm{\phi}_V(\bm{x})+{u})+\bm{d}_{u}
\end{split}
\end{equation}
where we have introduced the \emph{lumped disturbance}:
\begin{align*}
    \bm{d}_u &= (\bm{A}-\bm{A}_{\rm r})\bm{x}
        +\bm{b}
            \Big(
                (\Lambda\bm{V}-\Lambda_{\rm r}\bm{V}_{\rm r})^T\bm{\phi}_V(\bm{x}) \\
                &\quad
                +(\Lambda-\Lambda_{\rm r}){u}+\Lambda\bm{W}^T\bm{\phi}_W(\bm{x})+\Lambda{d}
            \Big) \numberthis\label{f:du}
\end{align*}

We assume the standard \emph{matching conditions} of MRAC that relate the unknown linear components in~(\ref{eq:NonlinearSystem_Disturbance}) to the reference model~(\ref{eq:ReferenceModel}):
\begin{asmp}[Matching conditions~\cite{Krstic.AdaptiveControl}] \label{asmp:MatchingCondition}
There exists a constant vector $\bm{k}^*$ and a constant scalar $k_{\rm r}^*$ such that
\begin{equation} \label{eq:MatchingCondition_LinearSys}
\begin{split}
    \bm{A}_{\rm r} &= \bm{A}+\bm{b}\Lambda{\bm{k}_x^*}^T \\
    \bm{b}\Lambda_{\rm r} &= \bm{b}{\Lambda}k_{\rm r}^*
\end{split}
\end{equation}
\end{asmp}

The proposed control framework is a variant of conventional MRAC that incorporates a \emph{disturbance-rejection input} $u_{\rm drj}$ yet to be designed:
\begin{align*}
    u &= [\hat{\bm{k}}_x^T,{\hat{k}_{\rm r}}][\bm{x}^T,r]^T
    -\hat{\bm{V}}^T\bm{\phi}_V(\bm{x})
    -\hat{\bm{W}}^T\bm{\phi}_W(\bm{x})
    +u_{\rm drj} \numberthis\label{eq:DesiredAdpativeControl_LinearSystem_Dist}
\end{align*}
where $u_{\rm drj}(t) \in \mathbb{R}$ is some bounded, piecewise continuously differentiable signal. If $u_{\rm drj}(t) = 0$, then this control law is the standard MRAC law. Define the \emph{adaptive parameters}
$\hat{\bm{k}}_x$, $\hat{k}_{\rm r}$, $\hat{\bm{V}}$, and $\hat{\bm{W}}$, which are governed by the following \emph{adaptation rules}:
\begin{equation} \label{eq:DistObserver_Summary_AdaptiveLaw}
\begin{split}
    &\quad
    \begin{cases}
        \begin{array}{rclrcl}
        \dot{\hat{\bm{k}}}_x &\!\!=\!\!& \bm{\Gamma}_x{\rm proj}(\hat{\bm{k}}_x,\bm{v}_x,f_x(\hat{\bm{k}}_x))
         \\
        \dot{\hat{k}}_{\rm r} &\!\!=\!\!& \gamma_{\rm r}{\rm proj}(\hat{k}_{\rm r},v_{\rm r},f_{\rm r}(\hat{k}_{\rm r})) \\
        \dot{\hat{\bm{V}}} &\!\!=\!\!& \bm{\Gamma}_{\rm V}{\rm proj}(\hat{\bm{V}},\bm{v}_{\rm V},f_{\rm V}(\hat{\bm{V}})) \\
        \dot{\hat{\bm{W}}} &\!\!=\!\!& \bm{\Gamma}_{\rm W}{\rm proj}(\hat{\bm{W}},\bm{v}_{\rm W},f_{\rm W}(\hat{\bm{W}}))
        \end{array}
    \end{cases}
\end{split}
\end{equation}
where
\begin{equation} \label{f:v}
\begin{split}
    &\bm{v}_x = -\bm{x}\bm{e}^T\bm{P}\bm{b}~\sgn{\Lambda},
    \quad
        v_{\rm r} = -r\bm{e}^T\bm{P}\bm{b}~\sgn{\Lambda} \\
    &\bm{v}_{\rm V} = \bm{\phi}_V(\bm{x})\bm{e}^T\bm{P}\bm{b}~\sgn{\Lambda},
    \quad
        \bm{v}_{\rm W} = \bm{\phi}_W(\bm{x})\bm{e}^T\bm{P}\bm{b}~\sgn{\Lambda}
\end{split}
\end{equation}
The matrices $\bm{\Gamma}_x$, $\gamma_{\rm r}$, $\bm{\Gamma}_{\rm V}$, $\bm{\Gamma}_{\rm W}$, and $\bm{P}$ in~\eqref{eq:DistObserver_Summary_AdaptiveLaw}--\eqref{f:v} are design parameters that are chosen to be positive-definite. The function `${\rm proj}(\cdot)$' is the projection operator that appears in the projection algorithm in~\cite{Pomet.ProjectionAlgorithm}, wherein convex functions  $f_x(\cdot)$, $f_{\rm r}(\cdot)$, $f_{\rm V}(\cdot)$, and $f_{\rm W}(\cdot)$ are introduced to confine the values of the adaptive parameters. Let $\tilde{\bm{k}}_x:=\hat{\bm{k}}_x-\bm{k}_x^*$, $\tilde{k}_{\rm r}:=\hat{k}_{\rm r}-k_{\rm r}^*$, $\tilde{\bm{V}}:=\hat{\bm{V}}-\bm{V}$, and $\tilde{\bm{W}}:=\hat{\bm{W}}-\bm{W}$. The projection algorithm ensures the existence of a time $T_{\rm p}$ such that
\begin{equation}
\|\tilde{\bm{k}}_x\| \leq b_{k_x}, \;
\|\tilde{\bm{k}}_{\rm r}\| \leq b_{k_{\rm r}}, \;
\|\tilde{\bm{V}}\| \leq b_{V}, \;
\|\tilde{\bm{W}}\| \leq b_{W}
\quad \forall t \geq T_p
\label{f:para_Performance}
\end{equation}
where $b_{k_x}$, $b_{k_{\rm r}}$, $b_{V}$, and $b_{W}$ are positive constants.

\begin{lem} \label{lem:Bounded_x}
Under the control law~\eqref{eq:DesiredAdpativeControl_LinearSystem_Dist}, the system state $\bm{x}$ remains bounded.
\end{lem}
\begin{proof}
Following the standard derivation of MRAC with the projection operator, the time derivative of the Lyapunov function
\begin{align*}
&V(\bm{e},\tilde{\bm{k}}_x,\tilde{k}_r,\tilde{\bm{V}},\tilde{\bm{W}})
    = \bm{e}^T\bm{P}\bm{e} \\
    &\quad
        +|\Lambda|
         \Big(
            \tilde{\bm{k}}_x^T\bm{\Gamma}_x^{-1}\tilde{\bm{k}}_x+\gamma_r^{-1}\tilde{k}_r^2 +\tilde{\bm{V}}^T\bm{\Gamma}_V^{-1}\tilde{\bm{V}}+\tilde{\bm{W}}^T\bm{\Gamma}_W^{-1}\tilde{\bm{W}}
            \Big) \numberthis\label{f:LyapunovFunction_Linear}
\end{align*}
is
\begin{align*}
    \dot{V}
        &=
            -\bm{e}^T\bm{Q}\bm{e}
            +
            2\bm{e}^T\bm{P}\bm{b}\Lambda(u_{\rm drj}+d) \\
        &\leq
            -\ubar{\lambda}_{Q}\norm{\bm{e}}^2
            +2\norm{\Lambda\bm{P}\bm{b}(u_{\rm drj}+d)}\norm{\bm{e}} \\
        &\leq
            0 \quad\forall\;\;\norm{\bm{e}}\geq\frac{2\norm{\Lambda\bm{P}\bm{b}(u_{\rm drj}+d)}}{\ubar{\lambda}_{Q}} \numberthis \label{f:DesiredPerformance}
\end{align*}
where $\bm{Q}$ is a positive definite matrix satisfying
$\bm{A}_{\rm r}^T\bm{P}+\bm{P}\bm{A}_{\rm r} = -\bm{Q}$ and $\ubar{\lambda}_{Q}$ is the minimum eigenvalue of $\bm{Q}$.
Because $u_{\rm drj}$ and $d$ are bounded, $\bm{e}$ is bounded. Since $\bm{x}_{\rm r}$ is bounded, as well, it follows that $\bm{x}$ is bounded.
\end{proof}

\section{The disturbance observer}
\label{sec:DisturbanceRejection}
It can be seen from~\eqref{f:DesiredPerformance} that if $u_{\rm drj}+d \approx 0$, it is possible the tracking error $\bm{e}$ will converge to a small value. Such a situation can be achieved by the concept of DOBC: utilizing the \emph{disturbance estimate}, denoted by $\hat{d}$, to counteract the disturbance $d$ by letting $u_{\rm drj}=-\hat{d}$. The proposed approach involves first evaluating a \emph{lumped disturbance estimate}, denoted by $\hat{\bm{d}}_u$, and then deducing $\hat{d}$ from it.

\begin{asmp} \label{asmp:LumpedDisturbanceObserver}
One has designed a \emph{lumped disturbance observer}:
\begin{align} \label{eq:LumpedDisturbanceObserver}
    \dot{\hat{\bm{d}}}_u = \bm{f}_{\rm DO}(\hat{\bm{d}}_u,\bm{y}), \quad\bm{y} = \bm{h}(\bm{x})
\end{align}
which satisfies
\begin{align*}
    \norm{\hat{\bm{d}}_u(t)-\bm{d}_u(t)}
    \leq
    \max{\left\{\beta\left(\norm{\hat{\bm{d}}_u(t_0)-\bm{d}_u(t_0)},t-t_0\right),\epsilon_{d_u}\right\}} \numberthis\label{eq:LumpedDisturbanceObserver_Performance}
\end{align*}
where $\beta$ is a class $\mathcal{KL}$ function \cite{Khalil.NonlinearSys} and $\epsilon_{d_u}$ is some small positive constant. The signal $\bm{y}$ is some measurement that is required by the chosen estimation method~\eqref{eq:LumpedDisturbanceObserver}.
\end{asmp}

Constructing a lumped disturbance observer as in Assumption~\ref{asmp:LumpedDisturbanceObserver} is straightforward. Various designs have been proposed (see~\cite{Li&Chen.DOBC}, for instance) and many guarantee the performance indicated by Lemma~\ref{lem:Estimator_DistTolerent}, wherein an exponentially convergent smooth observer ensures boundedness of the estimate error as long as the disturbance is bounded.
\begin{lem} \label{lem:Estimator_DistTolerent}
Consider the system
\begin{align*}
    \dot{\bm{x}}^* &= \bm{f}(\bm{x}^*,\bm{u})
\end{align*}
and the observer
\begin{align*}
    \dot{\hat{\bm{x}}}^* = \bm{f}_{\rm o}(\hat{\bm{x}}^*,\bm{u},\bm{h}(\bm{x}^*))
\end{align*}
where $\bm{f}$ and $\bm{f}_{\rm o}$ are continuously differentiable. Suppose
\begin{align*}
    \norm{\hat{\bm{x}}^*(t)-\bm{x}^*(t)} \leq k_1\norm{\hat{\bm{x}}^*(t_0)-\bm{x}^*(t_0)}e^{-\lambda_1(t-t_0)}
\end{align*}
where $k$ and $\lambda$ are positive constants. Then, for the system
\begin{align*}
    \dot{\bm{x}} &= \bm{f}(\bm{x},\bm{u})+\bm{d}
\end{align*}
where $\norm{\bm{d}} \leq b_d$ for some positive constant $b_d$, the observer
\begin{align*}
    \dot{\hat{\bm{x}}} = \bm{f}_{\rm o}(\hat{\bm{x}},\bm{u},\bm{h}(\bm{x}))
\end{align*}
gives
\begin{align*}
    \norm{\hat{\bm{x}}(t)-\bm{x}(t)} \leq \max{\left\{k_2\norm{\hat{\bm{x}}(t_0)-\bm{x}(t_0)}e^{-\lambda_2(t-t_0)},\epsilon\right\}}
\end{align*}
where $k_2$, $\lambda_2$, and $\epsilon$ are positive constants.
\end{lem}
\begin{proof}
Let $\bm{e}^*:=\hat{\bm{x}}^*-\bm{x}^*$. By~\cite[Thm.~4.14]{Khalil.NonlinearSys}, there is a differentiable function $V(\bm{e}^*,t)$ satisfying
\begin{gather*}
    c_1\norm{\bm{e}^*}^2 \leq V(\bm{e}^*,t) \leq c_2\norm{\bm{e}^*}^2 \\
    \frac{\partial V}{\partial t}+\frac{\partial V}{\partial \bm{e}^*}\Big(\bm{f}_{\rm o}(\hat{\bm{x}}^*,\bm{u},\bm{h}(\bm{x}^*))-\bm{f}(\bm{x}^*,\bm{u})\Big) \leq -c_3\norm{\bm{e}^*}^2 \\
    \norm{\frac{\partial V}{\partial \bm{e}^*}} \leq c_4\norm{\bm{e}^*}
\end{gather*}
where $c_1$, $c_2$, and $c_3$ are positive constants. Let $\bm{e}:=\hat{\bm{x}}-\bm{x}$. We have
\begin{align*}
    \dot{V}(\bm{e},t) &=
    \frac{\partial V}{\partial t}+\frac{\partial V}{\partial \bm{e}}\left(\bm{f}_{\rm o}(\hat{\bm{x}},\bm{u},\bm{h}(\bm{x}))-\bm{f}(\bm{x},\bm{u})\right)
        -\frac{\partial V}{\partial \bm{e}}\bm{d} \\
    &\leq
        -c_3\norm{\bm{e}}^2+c_4\norm{\bm{e}}b_d
        =
        \left(-c_3+\frac{c_4 b_d}{\norm{\bm{e}}}\right)\norm{\bm{e}}^2
\end{align*}
and
\begin{align} \label{f:Vdot_quadratic}
    \dot{V}(\bm{e},t)
    \leq
        \left(-c_3+\frac{c_4 b_d}{\norm{\bm{e}}}\right)\norm{\bm{e}}^2
    <
        0 \quad\forall\; \norm{\bm{e}} > \frac{c_4 b_d}{c_3}
\end{align}
Relationship~\eqref{f:Vdot_quadratic} implies
\begin{align*}
    -c_3+\frac{c_4 b_d}{\norm{\bm{e}}} < 0
    \quad\forall\; \norm{\bm{e}} > \frac{c_4 b_d}{c_3}
\end{align*}
Thus, there exists a positive value $M_e$ where $M_e > \frac{c_4 b_d}{c_3}$ such that, for $\norm{\bm{e}} \geq M_e$,
\begin{align*}
    \dot{V}(\bm{e},t)
    \leq
        \left(-c_3+\frac{c_4 b_d}{\norm{\bm{e}}}\right)\norm{\bm{e}}^2
    \leq
        \left(-c_3+\frac{c_4 b_d}{M_e}\right)\norm{\bm{e}}^2
    <
        -c^*\norm{\bm{e}}^2
\end{align*}
where $c^*$ is a positive constant. By \cite[Thm~4.10]{Khalil.NonlinearSys}, we obtain
\begin{align*}
    \forall\;\norm{\bm{e}} \geq M_e, \quad
    \norm{\bm{e}(t)} \leq
        \left(\frac{c_2}{c_1}\right)^{\frac{1}{2}}\norm{\bm{e}(t_0)}e^{-\frac{c^*}{2 c_2}(t-t_0)}
\end{align*}
Taking $k_2 = \left(\frac{c_2}{c_1}\right)^{\frac{1}{2}}$, $\lambda_2 = -\frac{c^*}{2 c_2}$, and $\epsilon = M_e$ completes the proof.
\end{proof}

We note that the lumped disturbance $\bm{d}_u$ given in~(\ref{f:du}) is bounded, since $\bm{x}$ and $d$ are bounded and $\bm{\phi}_W(\cdot)$ and $\bm{\phi}_V(\cdot)$ are piecewise differentiable.

Let $\bm{e}_{d_u} := \hat{\bm{d}}_u-\bm{d}_u$. Assumption~\ref{asmp:LumpedDisturbanceObserver} says there exists a time $T_{e_{d_u}}$ such that
\begin{align} \label{f:edu_Performance}
    \norm{\bm{e}_{d_u}} \leq \epsilon_{d_u} \quad\forall\; t \geq T_{e_{d_u}}
\end{align}
Recognize that $\hat{d}$ and $\hat{\bm{d}}_u$ satisfy the expression
\begin{align*}
        \bm{b}\Lambda_{\rm r}\hat{d} &= \hat{\bm{d}}_u +\bm{b}\Lambda_{\rm r}\Big(\hat{\bm{k}}_x^T\bm{x}+(\hat{k}_r-1)r \\
        &\quad
        -(\hat{\bm{V}}-\bm{V}_{\rm r})^T\bm{\phi}_V(\bm{x})-\hat{\bm{W}}^T\bm{\phi}_W(\bm{x})\Big) \numberthis\label{eq:DistObserver_Summary_DisturbanceEstimate}
\end{align*}
Let $e_d := \hat{d}-d$. By~\eqref{eq:DistObserver_Summary_DisturbanceEstimate},~\eqref{f:du}, and~\eqref{eq:DesiredAdpativeControl_LinearSystem_Dist}, we obtain
\begin{align*}
    &\bm{b}\Lambda{e_d} = \bm{b}\Lambda_{\rm r}(\hat{d}-d)+\bm{b}(\Lambda-\Lambda_{\rm r})(\hat{d}-d) \\
    &=
    \hat{\bm{d}}_u + \bm{b}\Lambda_{\rm r}\Big(\hat{\bm{k}}_x^T\bm{x}+({\hat{k}_r}-1)r
            -(\hat{\bm{V}}-\bm{V}_{\rm r})^T\bm{\phi}_V(\bm{x})
            -\hat{\bm{W}}^T\bm{\phi}_W(\bm{x})\Big) \\
            &\quad
            +\bm{b}(\Lambda-\Lambda_{\rm r})\hat{d}-\bm{b}\Lambda{d} \\
    &=
    (\bm{A}-\bm{A}_{\rm r})\bm{x}
        +\bm{b}\Big(
                (\Lambda\bm{V}-\Lambda_{\rm r}\bm{V}_{\rm r})^T\bm{\phi}_V(\bm{x})
                    +\Lambda\bm{W}^T\bm{\phi}_W(\bm{x})
            \Big)
        +\bm{e}_{d_u} \\
        &\quad
            +\bm{b}(\Lambda-\Lambda_{\rm r})
            \Big(
                \hat{\bm{k}}_x^T\bm{x}+{\hat{k}_{\rm r}}r
                -\hat{\bm{V}}^T\bm{\phi}_V(\bm{x})
                -\hat{\bm{W}}^T\bm{\phi}_W(\bm{x})
                +u_{\rm drj}+\hat{d}
            \Big) \\
        &\quad
            +\bm{b}\Lambda_{\rm r}\Big(\hat{\bm{k}}_x^T\bm{x}+({\hat{k}_r}-1)r
                -(\hat{\bm{V}}-\bm{V}_{\rm r})^T\bm{\phi}_V(\bm{x})
                -\hat{\bm{W}}^T\bm{\phi}_W(\bm{x})\Big)
\end{align*}
Finally, by~\eqref{eq:MatchingCondition_LinearSys}, we find
\begin{align} \label{eq:ed_equation}
    \bm{b}\Lambda{e_d} = \bm{b}\Lambda\tilde{u}_{\rm adp} +\bm{e}_{d_u}  + \bm{b}(\Lambda-\Lambda_{\rm r})\eta
\end{align}
where
\begin{align}
    \tilde{u}_{\rm adp} &:=
    \tilde{\bm{k}}_x^T\bm{x}+\tilde{k}_{\rm r}r-\tilde{\bm{V}}^T\bm{\phi}_V(\bm{x})-\tilde{\bm{W}}^T\bm{\phi}_W(\bm{x}) \label{f:u_adp} \\
    \eta &:= u_{\rm drj}+\hat{d} \label{f:eta}
\end{align}
The fact that $\bm{x}$ and $r$ are bounded and $\bm{\phi}_W(\cdot)$ and $\bm{\phi}_V(\cdot)$ are piecewise differentiable ensures $\tilde{u}_{\rm adp}$ is bounded. Because $u_{\rm drj}$ is bounded as well, we have
\begin{align} \label{f:eta_Performance}
    |\eta| \leq \epsilon_{\eta} \quad\forall\; t \geq T_{\eta}
\end{align}
where $\epsilon_{\eta}$ is a finite constant and $T_{\eta}$ is some finite convergence time. Along with~\eqref{f:para_Performance}, ~\eqref{f:edu_Performance}, and~\eqref{f:eta_Performance}, there exists a $T \geq \max{\{T_{\rm p},T_{e_{d_u}},T_{\eta}\}}$ such that
\begin{align} \label{f:ed_Performance}
    |e_d| \leq b_{e_d} \quad\forall\; t \geq T
\end{align}
where
\begin{align} \label{f:b_ed}
    b_{e_d} := \beta_{\rm adp}(\bm{x},r) + \frac{1}{|\Lambda|\left(\bm{b}^T\bm{b}\right)}
    \left(\norm{\bm{b}}\epsilon_{d_u}+|\Lambda-\Lambda_{\rm r}|\epsilon_{\eta}\right)
\end{align}
wherein
\begin{align} \label{f:b_utilde_adp}
    \beta_{\rm adp}(\bm{x},r) :=
    b_{k_x}\norm{\bm{x}}+b_{k_{\rm r}}|r|+b_{V}\norm{\bm{\phi}_V(\bm{x})}+b_{W}\norm{\bm{\phi}_W(\bm{x})}
\end{align}
which is an upper bound of~\eqref{f:u_adp}.

The value of $\beta_{\rm adp}$ can be quite small if we can tightly bound the adaptive parameters within a small neighborhood of their true values (i.e., $\bm{k}_x^*$, $\bm{k}_{\rm r}^*$, $\bm{V}$, and $\bm{W}$) using the projection algorithm. Of course, this depends on our knowledge of the uncertainty in the system parameters (i.e., $\bm{A}$, $\bm{V}$, $\bm{W}$, and $\Lambda$). The values
$$
    \epsilon_{d_u} \quad\text{and}\quad \epsilon_{\eta}
$$
appearing in the disturbance estimate error bound~(\ref{f:b_ed}) are determined by our design of the lumped disturbance observer~\eqref{eq:LumpedDisturbanceObserver} and the disturbance-rejection term $u_{\rm drj}$. The value $\epsilon_{d_u}$ is the \emph{lumped disturbance estimate convergence radius} in~(\ref{f:edu_Performance}). The parameter $\epsilon_{\eta}$ is the \emph{disturbance estimate rejection radius} that will be determined by our choice of $u_{\rm drj}$ whose purpose is to cancel the estimated disturbance $\hat{d}$.

\section{A Disturbance Estimation Tradeoff}
\label{sec:DifficultyInDisturbanceEstimation}
Following the last comment in the preceding section, one may want to choose
\begin{align} \label{eq:umod_Static}
    u_{\rm drj}=-\hat{d}
\end{align}
which according to~\eqref{f:eta} leads to $\epsilon_{\eta}=0$. Choosing $u_{\rm drj}$ as in~\eqref{eq:umod_Static}, however, can interfere with the disturbance estimation process; the lumped disturbance estimate convergence radius $\epsilon_{d_u}$ can become large. To see this, consider the dynamics of the lumped disturbance $\bm{d}_u$:
\begin{align*}
    \dot{\bm{d}}_u &= (\bm{A}-\bm{A}_{\rm r})\dot{\bm{x}}
        +\bm{b}
            \Big(
                (\Lambda\bm{V}-\Lambda_{\rm r}\bm{V}_{\rm r})^T\dot{\bm{\phi}}_V(\bm{x}) \\
                &\quad
                +(\Lambda-\Lambda_{\rm r})\dot{u}+\Lambda\bm{W}^T\dot{\bm{\phi}}_W(\bm{x})+\Lambda\dot{d}
            \Big)
\end{align*}
If $\bm{d}_u$ varies rapidly, the lumped disturbance observer will require high gains to generate an accurate ultimate estimate \cite{Li&Chen.DOBC}. The most likely source of fast-varying components is the input dynamics. Due to the input gain uncertainty (i.e., the unknown value $\Lambda-\Lambda_{\rm r}$), $\dot{\bm{d}}_u$ depends on $\dot{u}$, which is directly influenced by the adaptation rules~\eqref{eq:DistObserver_Summary_AdaptiveLaw} and by the disturbance observer dynamics
through the disturbance rejection input~\eqref{eq:umod_Static}. The value $\abs{\dot{u}}$ is typically large because adaptive laws and observers are usually quite aggressive for fast convergence.

Aside from adjusting the gains inherent in the adaptation laws and the disturbance observer to make the dynamics less aggressive (sacrificing fast convergence), one may design a ``slow-varying'' $u_{\rm drj}$, which still aims to cancel the (estimated) disturbance as in~\eqref{eq:umod_Static}, but which avoids creating fast variations in the control input $u$. Although the disturbance estimate rejection radius $\epsilon_{\eta}$ may be larger in this case, the lumped disturbance estimate convergence radius $\epsilon_{d_u}$ may be reduced so that the comprehensive ultimate bound~\eqref{f:b_ed} is improved. In other words, there is a trade-off when choosing design parameters between reducing $\epsilon_{d_u}$ and $\epsilon_{\eta}$. Managing this trade-off may be aided by a time scale separation between fast estimation of the lumped disturbance and slower rejection of the estimated disturbance.

\section{Design of the disturbance-rejection term}
\label{sec:IntuitiveDisturbanceRejection}
Instead of~\eqref{eq:umod_Static}, the disturbance rejection term is redefined by a magnitude- and rate-limited integral action with reset, wherein the integral action will be constructed based on the error $\eta$ in~\eqref{f:eta}. Let $f_{\rm drj}(t)$ represent the integrand of the integral action, a term limiting the changing rate of $u_{\rm drj}(t)$ that is yet to be designed, and let $\bar{u}_{\rm drj} > 0$ represent the magnitude limit on $u_{\rm drj}$. Formally, we write
\begin{equation} \label{eq:umod_Dynamical}
u_{\rm drj}(t) =
        \begin{cases}
            \int_{t^-}^{t}f_{\rm drj}(\tau) d\tau
            &\text{if $|u_{\rm drj}(t^-)| < \bar{u}_{\rm drj}$} \\
            -\hat{d}(t)
            &\text{if $|u_{\rm drj}(t^-)| \geq \bar{u}_{\rm drj}$ and $|\hat{d}(t)| < \bar{u}_{\rm drj}$} \\
            0
            &\text{if $|u_{\rm drj}(t^-)| \geq \bar{u}_{\rm drj}$ and $|\hat{d}(t)| \geq \bar{u}_{\rm drj}$}
        \end{cases}
\end{equation}
where $t^- < t$ is the most recent time at which the conditions in~\eqref{eq:umod_Dynamical} were evaluated. In practice, the conditions are continuously evaluated and the unsaturated disturbance rejection control command in~\eqref{eq:umod_Dynamical} -- i.e., the top line on the right --is determined by a dynamic equation. When this computed command saturates, the input switches to simply cancel the estimated disturbance, provided this disturbance estimate is within the saturation limit. If both the computed command and the disturbance estimate exceed the saturation limit, then the computed command resets to zero. Both the second and the third cases re-trigger the first (unsaturated) case. Finally, note that according to Lemma~\ref{lem:Bounded_x}, the system state remains bounded during such a reset mechanism.

To design $f_{\rm drj}$, consider the time derivative of $\eta$ in~\eqref{f:eta} for the unsaturated case, i.e., the top-most case in~\eqref{eq:umod_Dynamical}:
\begin{align} \label{f:etadot_Original}
    \dot{\eta} = f_{\rm drj}+\dot{\hat{d}}
\end{align}
By taking the derivative of~\eqref{eq:DistObserver_Summary_DisturbanceEstimate}, the term $\dot{\hat{d}}$ satisfies
\begin{align*}
    &\bm{b}\Lambda_{\rm r}\dot{\hat{d}} = \bm{f}_{\rm DO}(\hat{\bm{d}}_u,\bm{y})
        +\bm{b}\Lambda_{\rm r}\Big(\dot{\hat{\bm{k}}}_x^T\bm{x}+{\dot{\hat{k}}_r}r
        -\dot{\hat{\bm{V}}}^T\bm{\phi}_V(\bm{x})-\dot{\hat{\bm{W}}}^T\bm{\phi}_W(\bm{x})\Big) \\
    &
        +\bm{b}\Lambda_{\rm r}\Big(\hat{\bm{k}}_x^T\dot{\bm{x}}+(\hat{k}_r-1)\dot{r}
        -(\hat{\bm{V}}-\bm{V}_{\rm r})^T\frac{\partial \bm{\phi}_V(\bm{x})}{\partial \bm{x}}\dot{\bm{x}}
        -\hat{\bm{W}}^T\frac{\partial \bm{\phi}_W(\bm{x})}{\partial \bm{x}}\dot{\bm{x}}\Big)
        \numberthis\label{f:dhatdot}
\end{align*}
Because $\dot{\bm{x}}$ appears in~\eqref{f:dhatdot}, $\dot{\hat{d}}$ is unknown (even though we know $\hat{d}$ by~\eqref{eq:DistObserver_Summary_DisturbanceEstimate}). Let $\dot{\hat{d}}^*$ denote the \emph{estimate} of $\dot{\hat{d}}$, which is defined by
\begin{align*}
    &\bm{b}\Lambda_{\rm r}\dot{\hat{d}}^* = \bm{f}_{\rm DO}(\hat{\bm{d}}_u,\bm{y})
        +\bm{b}\Lambda_{\rm r}\Big(\dot{\hat{\bm{k}}}_x^T\bm{x}+{\dot{\hat{k}}_r}r
            -\dot{\hat{\bm{V}}}^T\bm{\phi}_V(\bm{x})-\dot{\hat{\bm{W}}}^T\bm{\phi}_W(\bm{x})\Big) \\
    &\quad
        +\bm{b}\Lambda_{\rm r}
        \Big(\hat{\bm{k}}_x^T\dot{\bm{x}}^*+(\hat{k}_r-1)\dot{r}
        -(\hat{\bm{V}}-\bm{V}_{\rm r})^T\frac{\partial \bm{\phi}_V(\bm{x})}{\partial \bm{x}}\dot{\bm{x}}^* \\
        &\qquad
        -\hat{\bm{W}}^T\frac{\partial \bm{\phi}_W(\bm{x})}{\partial \bm{x}}\dot{\bm{x}}^*\Big) \numberthis\label{f:dhatstardot}
\end{align*}
where
\begin{align} \label{f:xdot_Star}
    \dot{\bm{x}}^* := \bm{A}_{\rm r}\bm{x} + \bm{b}\Lambda_{\rm r}(\bm{V}_{\rm r}^T\bm{\phi}_V(\bm{x})+{u})+\hat{\bm{d}}_{u}
\end{align}
The definition in~\eqref{f:xdot_Star} can be regarded as an estimate of $\dot{\bm{x}}$, which is computed by the equation~\eqref{eq:NonlinearSystem_Disturbance_Rewritten} wherein the unknown lumped disturbance $\bm{d}_u$ is replaced by the estimate $\hat{\bm{d}}_u$. Letting $e_{\dot{\hat{d}}} := \dot{\hat{d}}^*-\dot{\hat{d}}$, we obtain
\begin{equation} \label{f:Err_DistEst}
    e_{\dot{\hat{d}}}
    =
    \left(\bm{b}^T\bm{b}\right)^{-1}\bm{b}^T
    \Big(
        \hat{\bm{k}}_x^T-(\hat{\bm{V}}-\bm{V}_{\rm r})^T\frac{\partial \bm{\phi}_V(\bm{x})}{\partial \bm{x}}
            -\hat{\bm{W}}^T\frac{\partial \bm{\phi}_W(\bm{x})}{\partial \bm{x}}
            \Big)\bm{e}_{d_u}
\end{equation}
By~\eqref{f:para_Performance},~\eqref{f:edu_Performance}, and the fact that $\bm{x}$ is bounded, there exists a $T_{e_{\dot{\hat{d}}}}$ satisfying $T_{e_{\dot{\hat{d}}}} \geq \max{\{T_{\rm p},T_{e_{d_u}}\}}$ such that
\begin{align} \label{f:edhatdot_Performance}
    |e_{\dot{\hat{d}}}| \leq b_{e_{\dot{\hat{d}}}} \quad\forall\; t \geq T_{e_{\dot{\hat{d}}}}
\end{align}
where
\begin{align*}
    b_{e_{\dot{\hat{d}}}} &:=
    \left(\bm{b}^T\bm{b}\right)^{-1}\norm{\bm{b}}
    \Big(
        \left(\norm{\bm{k}_x^*}+b_{k_x}\right)
        +\left(b_V+\norm{\bm{V}-\bm{V}_{\rm r}}\right)\norm{\frac{\partial \bm{\phi}_V(\bm{x})}{\partial \bm{x}}} \\
        &\quad
        +
        \left(\norm{\bm{W}}+b_{W}\right)\norm{\frac{\partial \bm{\phi}_W(\bm{x})}{\partial \bm{x}}}
    \Big)\epsilon_{d_u}
\end{align*}
Referring to Assumption~\ref{asmp:LumpedDisturbanceObserver}, the value of $\epsilon_{d_u}$ can be kept quite small by properly designing the lumped disturbance observer.

We are now ready to define the integrand $f_{\rm drj}$ for the integral action term in the disturbance rejection control law~\eqref{eq:umod_Dynamical}. To prevent $\|\dot{\bm{d}}_u\|$ from going to an extreme value due to a large $|f_{\rm drj}|$, we anticipate
\begin{align} \label{f:f_drj}
    f_{\rm drj} =
        \begin{cases}
            \varphi_{\rm drj}
            &\text{if $|\varphi_{\rm drj}| < \bar{f}_{\rm drj}$} \\
            \sgn{\varphi_{\rm drj}}\bar{f}_{\rm drj}         &\text{if $|\varphi_{\rm drj}| \geq \bar{f}_{\rm drj}$}
        \end{cases}
\end{align}
where $\varphi_{\rm drj}$ is a function to be designed and $\bar{f}_{\rm drj}$ is a prescribed limit on $|f_{\rm drj}|$, i.e., a rate limit on $u_{\rm drj}$. Let
\begin{align} \label{f:udotStar}
    \varphi_{\rm drj} := -k_{\eta}\eta-\dot{\hat{d}}^*
\end{align}
for some positive constant $k_{\eta}$. Substituting~\eqref{f:f_drj} into~\eqref{f:etadot_Original} with $\varphi_{\rm drj}$ given in~\eqref{f:udotStar}, we obtain
\begin{align*}
    \dot{\eta} &= -k_{\eta}\eta+d_{\eta} \numberthis\label{f:etahatdot} \\
    d_{\eta} &=
        \begin{cases}
            -e_{\dot{\hat{d}}}
                &\text{if $|\varphi_{\rm drj}| < \bar{f}_{\rm drj}$} \\
            -e_{\dot{\hat{d}}}+\delta \varphi_{\rm drj}
                &\text{if $|\varphi_{\rm drj}| \geq \bar{f}_{\rm drj}$}
        \end{cases}
\end{align*}
where
$
    \delta \varphi_{\rm drj} = \sgn{\varphi_{\rm drj}}\bar{f}_{\rm drj}-\varphi_{\rm drj}
$. The function $V_{\eta}(\eta) = \frac{1}{2}\eta^2$ shows that
\begin{align} \label{f:LypunovFunciton_etahat}
    \dot{V}_{\eta}(\eta)
    = -k_{\eta}\eta^2 + \eta d_{\eta}
    \leq 0 \qquad\forall\; |\eta| \geq \frac{|d_{\eta}|}{k_{\eta}}
\end{align}
If $|\varphi_{\rm drj}| < \bar{f}_{\rm drj}$, input-to-state stability of~\eqref{f:etahatdot} ensures that $\eta$ is bounded since $e_{\dot{\hat{d}}}$ is bounded. In this case, we have $\epsilon_{\eta} = b_{e_{\dot{\hat{d}}}}/k_{\eta}$ and $T_{\eta} = T_{e_{\dot{\hat{d}}}}$. If $|\varphi_{\rm drj}| \geq \bar{f}_{\rm drj}$, then~\eqref{f:etahatdot} suggests $\eta$ might diverge because the term $\delta \varphi_{\rm drj}$ has not been shown to be remain bounded. However, the second and third cases in~\eqref{eq:umod_Dynamical} ensure that $u_{\rm drj}$ is reset to a value such that $|u_{\rm drj}| \leq \bar{u}_{\rm drj}$, meaning $\eta$ is also reset by~\eqref{f:eta} and remains bounded.

System~\eqref{f:etahatdot} can be regarded as a low-pass filter with $d_{\eta}$ as the input and $\eta$ as the output, which attenuates the high-frequency components of $d_{\eta}$ that come from the lumped disturbance estimate error $\bm{e}_{d_u}$ in~\eqref{f:Err_DistEst}. Although a large $k_{\eta}$ could potentially induce a large $\eta$ due to the \emph{peaking phenomenon} \cite[Thm~11.4]{Khalil.NonlinearSys}, the saturation on $u_{\rm drj}$ prevents severe excursions. The parameter $k_{\eta}$ provides additional freedom to adjust the control performance. The simulations in Section~\ref{sec:Example} reveal that the disturbance rejection with the integral action~\eqref{eq:umod_Dynamical} can outperform the intuitive design~\eqref{eq:umod_Static}.

\section{Performance analysis}
\label{sec:PerformanceAnalysis}
The condition concerning whether the magnitude of $u_{\rm drj}$ exceeds $\bar{u}_{\rm drj}$ determines one of two operating modes:
\begin{enumerate}
    \item \textbf{Performance enhancement.}
        When $u_{\rm drj}$ is \emph{not} saturated, the control input attempts to counteract disturbance $d$.
    \item \textbf{Robust stabilization.}
        When $u_{\rm drj}$ \emph{is} saturated, the control input is the standard MRAC law.
\end{enumerate}
Robustness of stability to the disturbance $d$ is guaranteed by Theorem~\ref{lem:Bounded_x}. In the performance improvement mode, $u_{\rm drj}$ actively seeks a value to counter the disturbance.

The following theorem indicates the \emph{best} performance provided by the proposed method.

\begin{thm} \label{thm:BestPerformance_Dynamical}
Let $\bm{A}_{\rm r}^T\bm{P}+\bm{P}\bm{A}_{\rm r} =: -\bm{Q}$ and let $\ubar{\lambda}_{Q}$ denote the minimum eigenvalue of $\bm{Q}$. If there exists a time $t^*$ such that $|u_{\rm drj}| < \bar{u}_{\rm drj}$ for $t \geq t^*$, the proposed DOBAC~\eqref{eq:DesiredAdpativeControl_LinearSystem_Dist}--\eqref{f:v} with~\eqref{eq:LumpedDisturbanceObserver},~\eqref{eq:umod_Dynamical}, and~\eqref{f:edhatdot_Performance}--\eqref{f:f_drj} guarantees
$$
    \epsilon_{\rm r}
        =\frac{2}{\ubar{\lambda}_{Q}}\left(\frac{b_{e_{\dot{\hat{d}}}}}{k_{\eta}}+b_{e_d}\right)\norm{\Lambda\bm{P}\bm{b}}
$$
in the tracking error condition~\eqref{f:trackingError_Performance}.
\end{thm}
\begin{proof}
From~\eqref{f:LypunovFunciton_etahat} with the fact that $d_{\eta}=-e_{\dot{\hat{d}}}$, we obtain
$\epsilon_{\eta}=\frac{b_{e_{\dot{\hat{d}}}}}{k_{\eta}}$. Definition~\eqref{f:eta} gives
$$
    \eta-e_d = u_{\rm drj}+d
$$
Consider the Lyapunov function~\eqref{f:LyapunovFunction_Linear} again. Following the standard derivation of MRAC with the projection operator, we obtain~\eqref{f:DesiredPerformance}. By~\eqref{eq:umod_Dynamical}, in the case where $u_{\rm drj}$ is not saturated, we have
\begin{align*}
    \dot{V}
        &\leq
            -\ubar{\lambda}_{Q}\norm{\bm{e}}^2
            +2\norm{\Lambda\bm{P}\bm{b}(\eta-e_d)}\norm{\bm{e}} \\
        &\leq
            -\ubar{\lambda}_{Q}\norm{\bm{e}}^2
            +2
            \left(\frac{b_{e_{\dot{\hat{d}}}}}{k_{\eta}}+b_{e_d}\right) \norm{\Lambda\bm{P}\bm{b}}\norm{\bm{e}}
                \quad\forall\; t \geq T \\
        &\leq
            0 \quad\forall\;\norm{\bm{e}}\geq\frac{2}{\ubar{\lambda}_{Q}}\left(\frac{b_{e_{\dot{\hat{d}}}}}{k_{\eta}}+b_{e_d}\right)\norm{\Lambda\bm{P}\bm{b}} \quad\forall\; t \geq T
\end{align*}
which proves the claim.
\end{proof}

The reason Theorem~\ref{thm:BestPerformance_Dynamical} represents the best performance is that it describes the outcome when the DOBAC operates in the performance enhancement mode. Otherwise, the system operates in the robust stabilization mode, which guarantees boundedness but does not aim to actively cancel the disturbance. In implementation, it would be difficult to know in advance how the control law would switch between these two modes; it is even possible that the time $t^*$ does not exist, meaning the controller must remain in robust stabilization mode. On the other hand, the existence of $t^*$ indicates the controller eventually settles into the mode of performance improvement.

\section{Example: Position control of a mass-spring-damper system with a nonlinear spring}
\label{sec:Example}
Consider the following system:
\begin{equation} \label{eq:1DoFSys}
\begin{split}
    \dot{x}_1 &= x_2 \\
    \dot{x}_2 &= - a_1x_1 - a_2{x}_1^3 - b{x}_2 + \Lambda(u + d)
\end{split}
\end{equation}
where
$
    d = 5\sin{(0.5t)}
$.
The system coefficients $a_1$, $a_2$, $b$, and $\Lambda$ are unknown constants, but we know that $-0.5 \leq a_1 \leq 1.5$, $-0.5 \leq a_2 \leq 1.5$, $-0.5 \leq b \leq 1.5$, and $1 \leq \Lambda \leq 1.4$. (The simulations adopt the values $a_1 = 0.5$, $a_2 = 0.5 $, $b = 0.5$, and $\Lambda = 1.2$.) The system can be rewritten as
\begin{equation} \label{eq:1DoFSys_ObserverDesign}
\begin{split}
    \dot{\bm{x}} = \bm{A}_{\rm r}\bm{x} + \bm{b}\Lambda_{\rm r}(V_{\rm r}^T\phi_V(\bm{x})+{u})+\bm{d}_{u}
\end{split}
\end{equation}
where $\bm{x} = [x_1,x_2]^T$, $\Lambda_{\rm r} = 1$, $V_{\rm r} = 1$, and
\begin{gather*}
    \bm{A}_{\rm r} =
    \left[
    \begin{array}{cc}
        0 & 1 \\
        -1 & -1
    \end{array}
    \right]
    ,\quad
    \bm{b} =
    \left[
    \begin{array}{c}
        0 \\
        1
    \end{array}
    \right]
    ,\quad
    \phi_V(\bm{x}) = x_1^3
\end{gather*}
and $\bm{d}_{u}$ represents the remaining terms in the $\bm{x}$ dynamics. In~\cite{Chen&Woolsey.TCST23}, a global disturbance observer was designed to estimate $\bm{d}_{u}$ for~\eqref{eq:1DoFSys_ObserverDesign}. Adopting that observer, and referring the reader to \cite{Chen&Woolsey.TCST23} for the details, Assumption~\ref{asmp:LumpedDisturbanceObserver} holds.  Comparing~\eqref{eq:1DoFSys_ObserverDesign} and~\eqref{eq:NonlinearSystem_Disturbance}, we obtain
$$
    \bm{A} =
    \left[
    \begin{array}{cc}
        0 & 1 \\
        -a_1 & -b
    \end{array}
    \right]
    ,\quad
    V = \frac{a_2}{\Lambda}
    ,\quad
    \phi_W(\bm{x}) = 0
$$
The control objective is to have $x_1$ track the signal $x_{1{\rm r}} := \sin{(t)}$. We use the reference model~\eqref{eq:ReferenceModel} with
$$
    r =
        \left[
        \begin{array}{cc}
            1 & 1
        \end{array}
        \right]
        \bm{x}_{\rm r}
        -\sin{(t)}
$$

Let $\bm{k}^*_x =: [k^*_{x1},k^*_{x2}]^T$. The matching condition~\eqref{eq:MatchingCondition_LinearSys} holds with
\begin{gather*}
    \frac{1}{1.4} \leq k^*_{\rm r} \leq \frac{1}{1}, \quad
    \frac{-1.5}{1.4} \leq k^*_1 \leq \frac{0.5}{1} \\
    \frac{-1.5}{1.4} \leq k^*_2 \leq \frac{0.5}{1}, \quad
    \frac{-0.5}{1.4} \leq V \leq \frac{1.5}{1}
\end{gather*}
Knowing the range of each unknown coefficient, we design the following three convex functions for the projection algorithm:
\begin{align*}
    f_{\rm r}(\theta) &= \alpha_{\rm r}(\theta-\theta_{\rm r}^*)^2+\bar{\theta}_{\rm r} \\
    f_{x}(\bm{\theta}) &= (\bm{\theta}-\bm{\theta}_{\rm x}^*)^T
    \left[
    \begin{array}{cc}
        \alpha_{x1} & 0 \\
        0 & \alpha_{x2}
    \end{array}
    \right]
    (\bm{\theta}-\bm{\theta}_x^*)
    +\bar{\bm{\theta}}_{x} \\
    f_{\rm V}(\theta) &= \alpha_{\rm V}(\theta-\theta_{\rm V}^*)^2+\bar{\theta}_{\rm V}
\end{align*}
where, with $\bm{\theta}_x^* =: [\theta_{x1}^*,\theta_{x2}^*]^T$ and $\bar{\bm{\theta}}_x =: [\bar{\theta}_{x1},\bar{\theta}_{x2}]^T$,
\begin{gather*}
    \theta_{\rm r}^* = \frac{1}{2}\left(\frac{1}{1.4}+\frac{1}{1}\right), \quad
    \theta_{x1}^* = \frac{1}{2}\left(\frac{-1.5}{1.4}+\frac{0.5}{1}\right) \\
    \theta_{x2}^* = \frac{1}{2}\left(\frac{-1.5}{1.4}+\frac{0.5}{1}\right), \quad
    \theta_{\rm V}^* = \frac{1}{2}\left(\frac{-0.5}{1.4}+\frac{1.5}{1}\right) \\
    \bar{\theta}_{\rm r} = \bar{\theta}_{x1} = \bar{\theta}_{x2} = \bar{\theta}_{\rm V} = -1
\end{gather*}
The coefficients $\alpha_{\rm r}$, $\alpha_{x1}$, $\alpha_{x2}$, and $\alpha_{\rm V}$ are chosen so that the sets $\Omega_{\rm r} := \{ \theta \;|\; f_{\rm r}(\theta) \leq 1 \}$, $\Omega_x := \{ \bm{\theta} \;|\; f_{x}(\bm{\theta}) \leq 1 \}$, and $\Omega_{\rm V} := \{ \theta \;|\; f_{\rm V}(\theta) \leq 1 \}$
satisfy
\begin{align*}
    \Omega_{\rm r} &\supseteq \Omega_{\rm r}^* := \left\{ \theta \;|\; \frac{1}{1.4} \leq \theta \leq \frac{1}{1} \right\} \\
    \Omega_x &\supseteq \Omega_x^* := \left\{ [\theta_1,\theta_2]^T \;|\;
            \frac{-1.5}{1.4} \leq \theta_1 \leq \frac{0.5}{1}, \; \frac{-1.5}{1.4} \leq \theta_2 \leq \frac{0.5}{1} \right\} \\
    \Omega_{\rm V} &\supseteq \Omega_{\rm V}^* := \left\{ \theta \;|\; \frac{-0.5}{1.4} \leq \theta \leq \frac{1.5}{1} \right\}
\end{align*}
Specifically in the simulations, with $s_{\rm r} := \frac{1}{2}\left(\frac{1}{1}-\frac{1}{1.4}\right)$, $s_{x1} := \frac{1}{2}\left(\frac{0.5}{1}-\frac{-1.5}{1.4}\right)$, $s_{x2} := \frac{1}{2}\left(\frac{0.5}{1}-\frac{-1.5}{1.4}\right)$, and $s_{\rm V} := \frac{1}{2}\left(\frac{1.5}{1}-\frac{-0.5}{1.4}\right)$,
\begin{align*}
    \Omega_{\rm r} &= \left\{ \theta \;|\; \theta_{\rm r}^*-s_{\rm r}-0.1 \leq \theta \leq \theta_{\rm r}^*+s_{\rm r}+0.1 \right\} \\
    \Omega_x &= \left\{ [\theta_1,\theta_2]^T \;|\;
            \theta_{x1}^*-s_{x1}-0.1 \leq \theta_1 \leq \theta_{x1}^*+s_{x1}+0.1, \right. \\
            &\quad\left.
            \theta_{x2}^*-s_{x2}-0.1 \leq \theta_2 \leq \theta_{x2}^*+s_{x2}+0.1 \right\} \\
    \Omega_{\rm V} &= \left\{ \theta \;|\; \theta_{\rm V}^*-s_{\rm V}-0.1 \leq \theta \leq \theta_{\rm V}^*+s_{\rm V}+0.1 \right\}
\end{align*}
The control law uses $\bar{u}_{\rm drj} = 10$ and $\bar{f}_{\rm drj} = 5$ with $\gamma_{\rm r} = 1$ and
\begin{gather*}
\bm{\Gamma}_x =
    \left[
    \begin{array}{cc}
        1 & 0 \\
        0 & 1
    \end{array}
    \right]
,\quad
\bm{\Gamma}_{\rm V} =
    \left[
    \begin{array}{cc}
        1 & 0 \\
        0 & 1
    \end{array}
    \right]
,\quad
\bm{P} =
    \left[
    \begin{array}{cc}
        1.5 & 0.5 \\
        0.5 & 1
    \end{array}
    \right]
\end{gather*}
in the adaptive law. We do not need to choose $\bm{\Gamma}_{\rm W}$ because the associated nonlinearity $\phi_W(\bm{x})$ is zero.

As a baseline for comparison, Figure~\ref{fig:x1Tracking_TypicalAC} displays the result with $u_{\rm drj}=0$, i.e., using conventional MRAC. Figure~\ref{fig:x1Tracking_DOBAC} shows the results when incorporating different disturbance observer-based designs of $u_{\rm drj}$.  Figure~\ref{subfig:x1Tracking_StaticDOBAC} shows the results of using the control law~\eqref{eq:umod_Static}, which we refer to as \emph{direct DOBAC} (D-DOBAC), while Figure~\ref{subfig:x1Tracking_DynamicalDOBAC} illustrates the results of using the control law of Theorem~\ref{thm:BestPerformance_Dynamical}, which we refer to as \emph{integrating DOBAC} (I-DOBAC) with $k_{\eta} = 1$. It can be seen that the two DOBACs outperform the conventional MRAC, and the I-DOBAC performs better than the D-DOBAC; a comparison of the tracking error is shown in Figure~\ref{subfig:TrackingError_Comparison}. Figure~\ref{subfig:TrackingError_epsr} shows the upper bound $\epsilon_{\rm r}$ on the tracking error, computed according to Theorem~\ref{thm:BestPerformance_Dynamical}. This upper bound is much larger than the true tracking error, which is attributed to our limited knowledge of the unknown coefficients: one can see from~\eqref{f:b_ed} that the bound $b_{e_d}$ on the estimate error of the disturbance is significantly affected by $\beta_{\rm adp}$ in~\eqref{f:b_utilde_adp}. Reducing the uncertainty in the unknown coefficients would reduce $\beta_{\rm adp}$, as well.

Figure~\ref{fig:du_ComparisonDOBAC} illustrates the idea of limiting the changing rate of the disturbance-rejection term $u_{\rm drj}$ to prevent the lumped disturbance from varying quickly. Using the same disturbance observer, system~\eqref{eq:1DoFSys} under D-DOBAC is subject to a lumped disturbance of sharper peaks than the system under I-DOBAC after the initial trainsient response. Recalling the discussion of $f_{\rm drj}$ around~\eqref{f:f_drj}--\eqref{f:etahatdot}, this amplitude reduction might be attributed to the saturation of $f_{\rm drj}$ that was introduced to prevent $\|\dot{\bm{d}}_u\|$ from going to extreme values. Figure~\ref{fig:d} shows that the disturbance estimate $\hat{d}$ is more accurate in the case of I-DOBAC. As mentioned in Section~\ref{sec:DifficultyInDisturbanceEstimation}, it is easier to estimate a slower-varying $\bm{d}_u$, which gives a smaller $\epsilon_{d_u}$ leading to a smaller $b_{e_d}$ by~\eqref{f:b_ed}. This reasoning may explain the improved accuracy of $\hat{d}$ under I-DOBAC.

Next, consider $u_{\rm drj}$ for each DOBAC. In Figure~\ref{fig:umod_Static}, it can be seen $u_{\rm drj}$ never exceeds the saturation value and during the transient period, when $\hat{d}$ is too large to be used by the D-DOBAC, $u_{\rm drj}$ is zero (Figure~\ref{subfig:umod_StaticDOBAC_ZoomIn}). For I-DOBAC, as well, $u_{\rm drj}$ never exceeds the saturation value. Further, its changing rate never exceeds the prescribed maximum $\bar{f}_{\rm drj}$; in Figure~\ref{subfig:umod_DynamicalDOBAC_ZoomIn} one can see that $f_{\rm drj}$ becomes saturated within the first $2$ seconds. Figure~\ref{subfig:eta} shows the time history of $\eta$; it converges to a small value which, according to Theorem~\ref{thm:BestPerformance_Dynamical}, is bounded by $b_{e_{\dot{\hat{d}}}}/k_{\eta}$.

\begin{figure}[h!]
    \centering
    \includegraphics[width=0.425\textwidth]{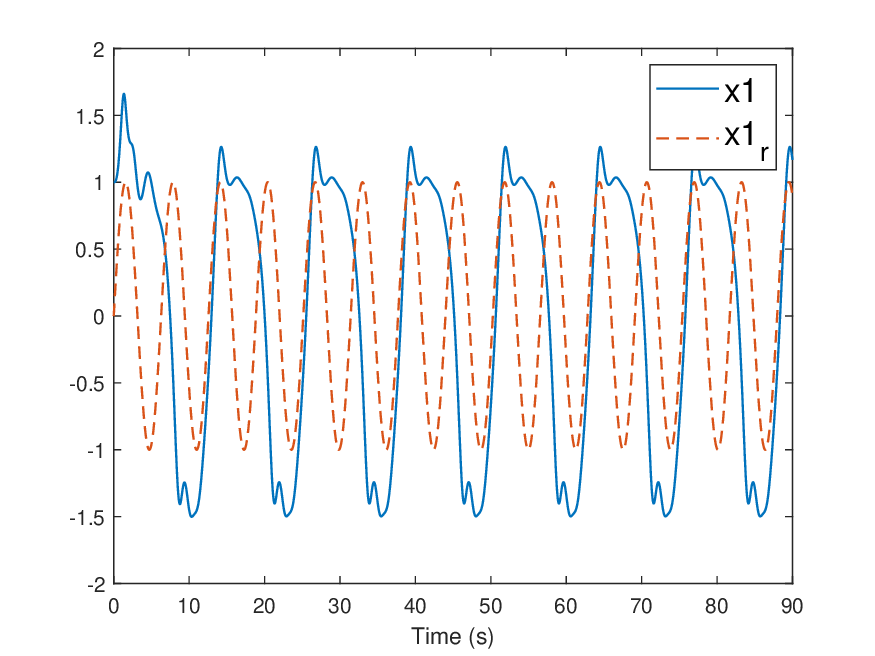}
    \caption{$x_1$ tracking with $u_{\rm drj}=0$. Dash: the desired trajectory $x_{1{\rm r}}$; solid: $x_1$.}
    \label{fig:x1Tracking_TypicalAC}
\end{figure}
\begin{figure}[h!]
\centering
    \begin{subfigmatrix}{2}
        \subfigure[Direct DOBAC]{
            \includegraphics[width=0.425\textwidth]{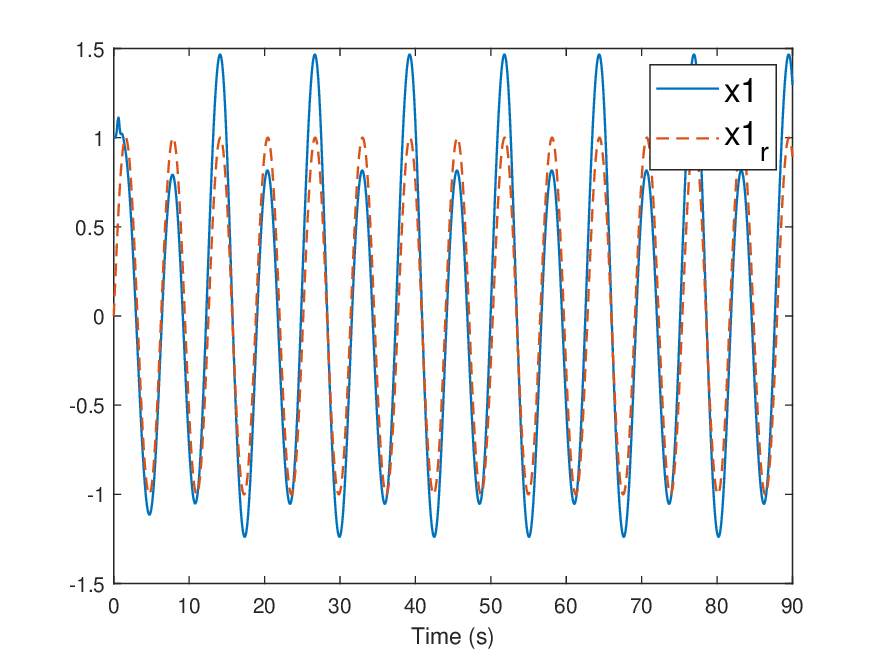}
            \label{subfig:x1Tracking_StaticDOBAC}}
        \subfigure[Integrating DOBAC]{
            \includegraphics[width=0.425\textwidth]{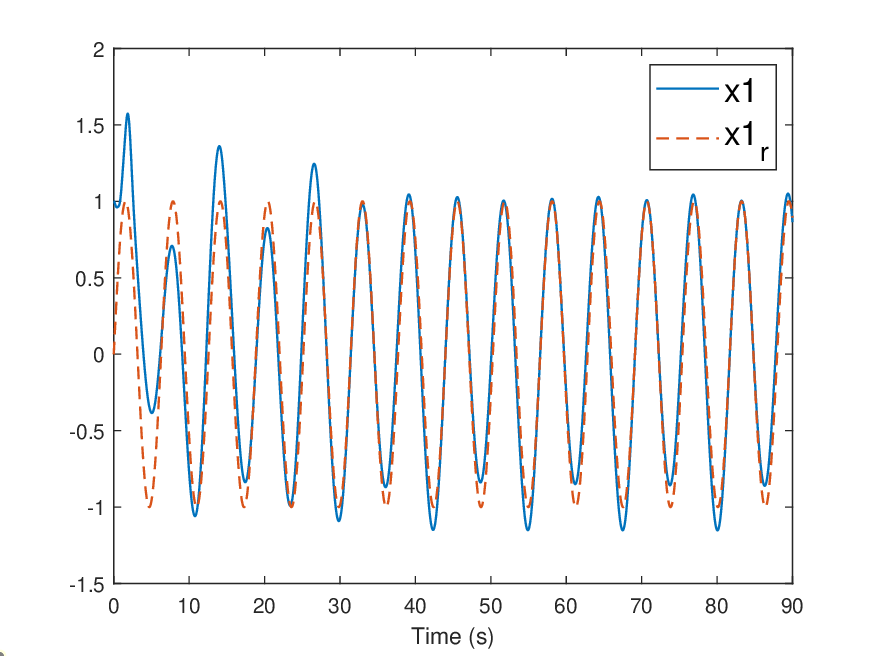}
            \label{subfig:x1Tracking_DynamicalDOBAC}}
    \end{subfigmatrix}
\caption{$x_1$ tracking with the DOBAC using different $u_{\rm drj}$. Dash: the desired trajectory $x_{1{\rm r}}$; solid: $x_1$.}
\label{fig:x1Tracking_DOBAC}
\end{figure}
\begin{figure}[h!]
\centering
    \begin{subfigmatrix}{2}
        \subfigure[Comparison of $\norm{\bm{e}}$ using different methods]{
            \includegraphics[width=0.425\textwidth]{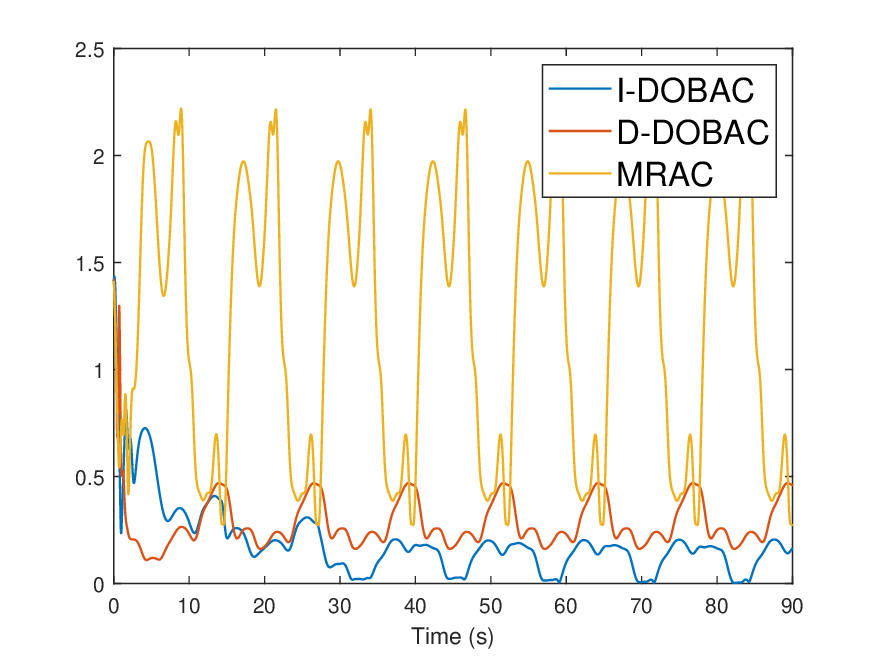}
            \label{subfig:TrackingError_Comparison}}
        \subfigure[$\norm{\bm{e}}$ using I-DOBAC, the upper bound $\epsilon_{\rm r}$, and $\beta_{\rm adp}$]{
            \includegraphics[width=0.425\textwidth]{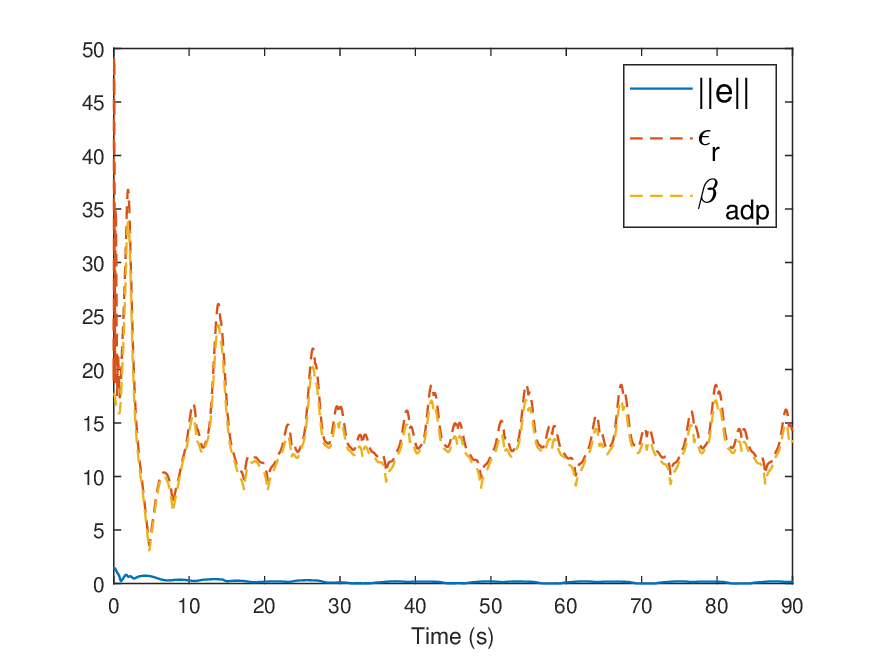}
            \label{subfig:TrackingError_epsr}}
    \end{subfigmatrix}
\caption{The tracking error history and the upper bound guaranteed by the integrating DOBAC}
\label{fig:TrackingError}
\end{figure}
\begin{figure}[h!]
    \centering
    \includegraphics[width=0.425\textwidth]{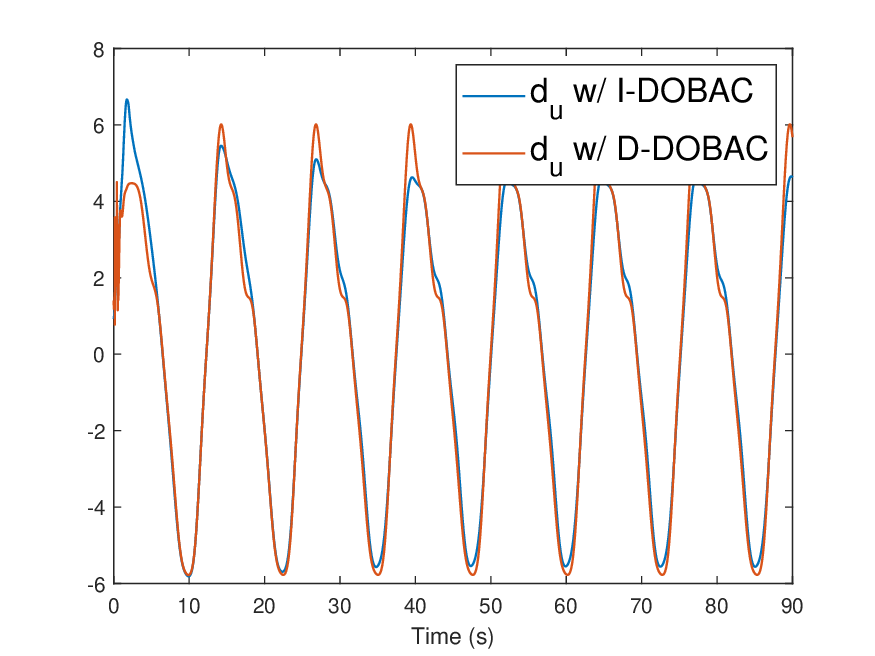}
    \caption{$\bm{d}_u$ with the DOBACs using different $u_{\rm drj}$}
    \label{fig:du_ComparisonDOBAC}
\end{figure}
\begin{figure}[h!]
\centering
    \begin{subfigmatrix}{2}
        \subfigure[Direct DOBAC]{
            \includegraphics[width=0.425\textwidth]{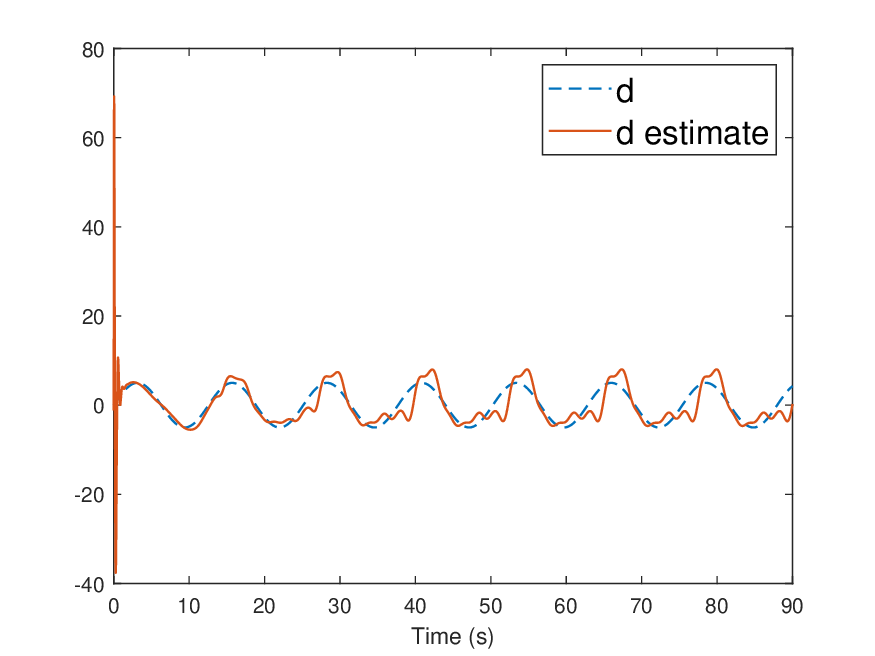}
            \label{subfig:d_StaticDOBAC}}
        \subfigure[Integrating DOBAC]{
            \includegraphics[width=0.425\textwidth]{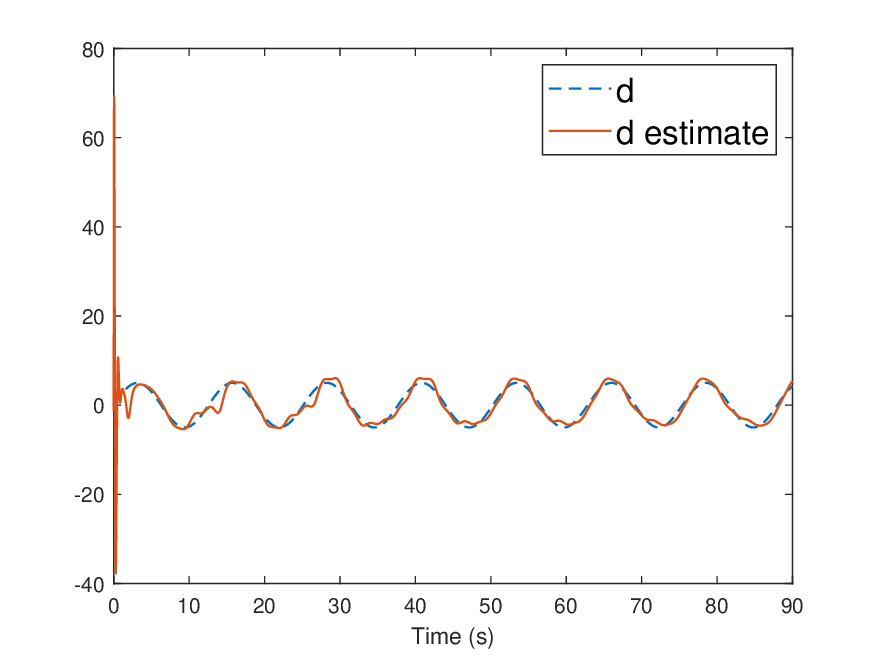}
            \label{subfig:d_DynamicalDOBAC}}
    \end{subfigmatrix}
\caption{Estimate of $d$ with the DOBAC using different $u_{\rm drj}$. Dash: $d$; solid: $\hat{d}$.}
\label{fig:d}
\end{figure}
\begin{figure}[h!]
\centering
    \begin{subfigmatrix}{2}
        \subfigure[]{
            \includegraphics[width=0.425\textwidth]{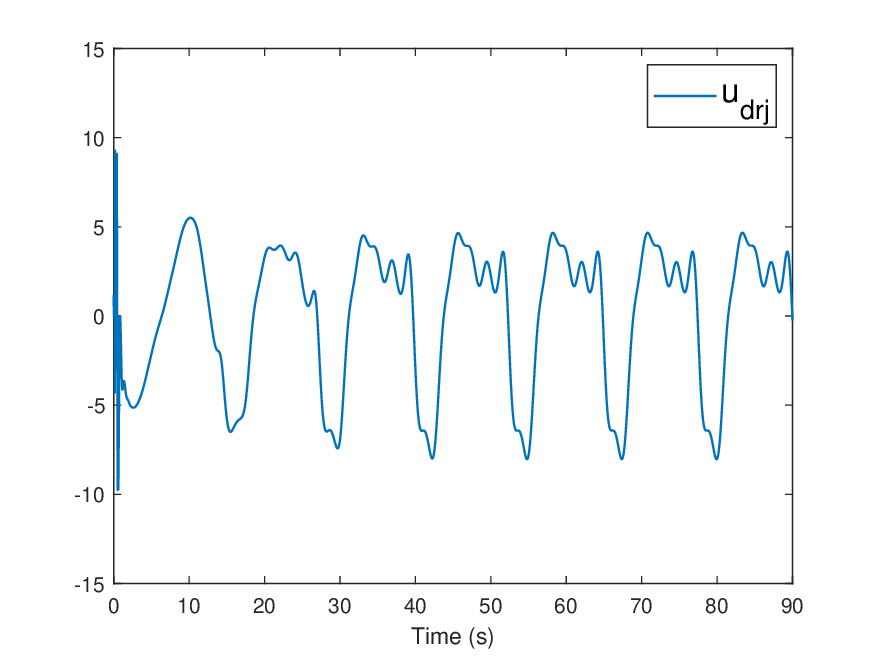}
            \label{subfig:umod_StaticDOBAC}}
        \subfigure[Zoom-in of Figure~\ref{subfig:umod_StaticDOBAC}]{
            \includegraphics[width=0.425\textwidth]{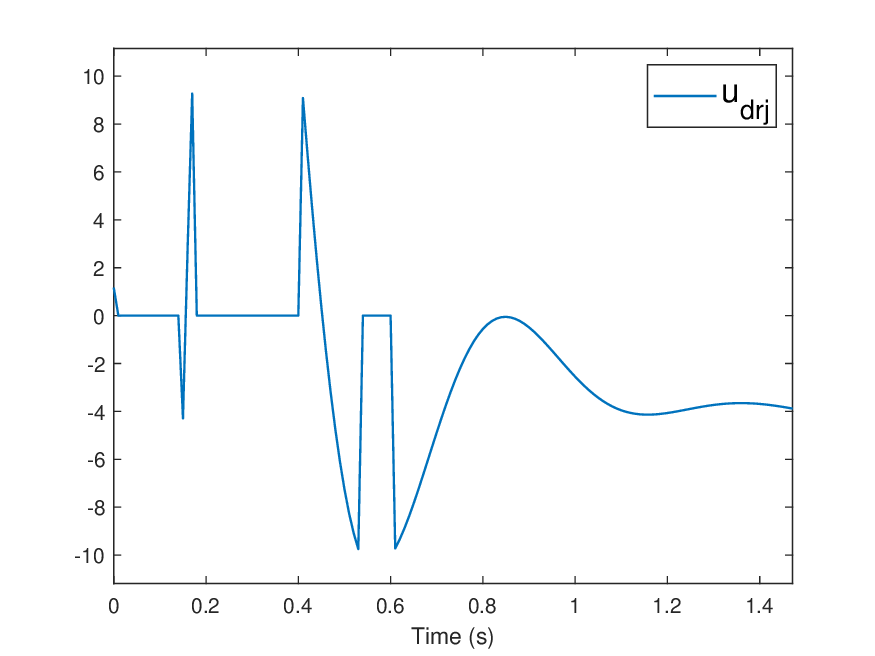}
            \label{subfig:umod_StaticDOBAC_ZoomIn}}
    \end{subfigmatrix}
\caption{$u_{\rm drj}$ of the direct DOBAC.}
\label{fig:umod_Static}
\end{figure}
\begin{figure}[h!]
\centering
    \begin{subfigmatrix}{3}
        \subfigure[]{
            \includegraphics[width=0.425\textwidth]{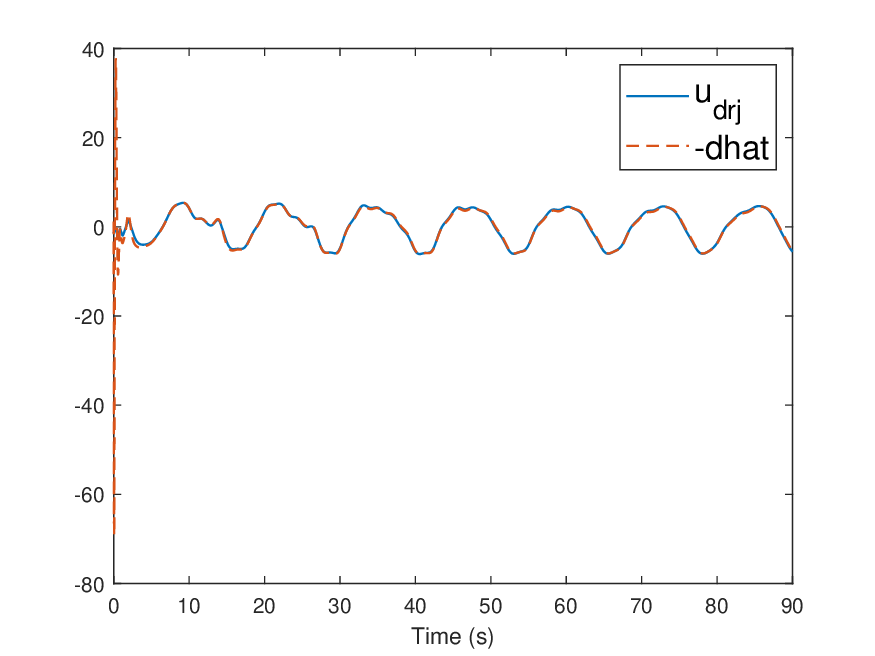}
            \label{subfig:umod_DynamicalDOBAC}}
        \subfigure[Zoom-in of Figure~\ref{subfig:umod_DynamicalDOBAC}]{
            \includegraphics[width=0.425\textwidth]{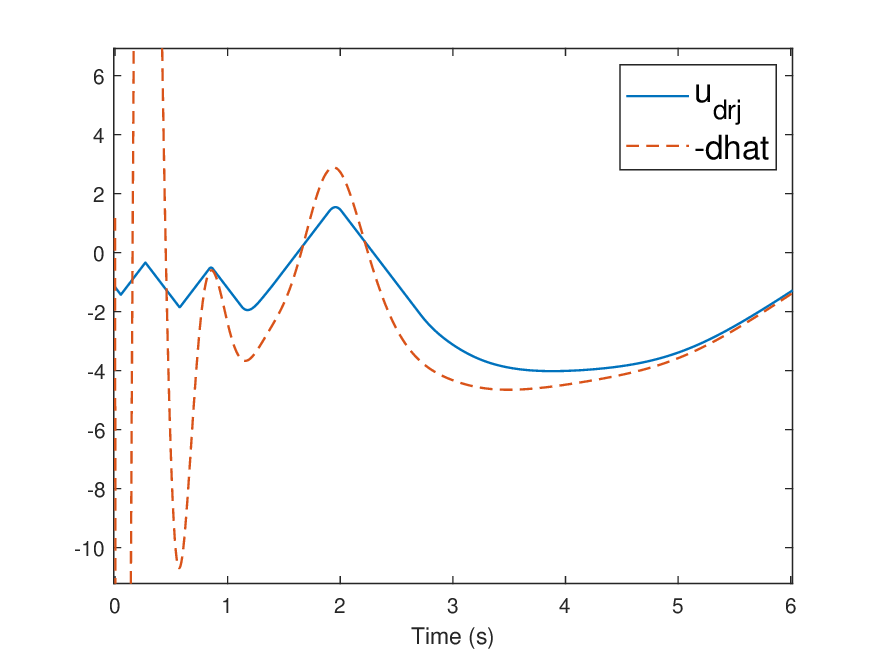}
            \label{subfig:umod_DynamicalDOBAC_ZoomIn}}
        \subfigure[$\eta$]{
            \includegraphics[width=0.425\textwidth]{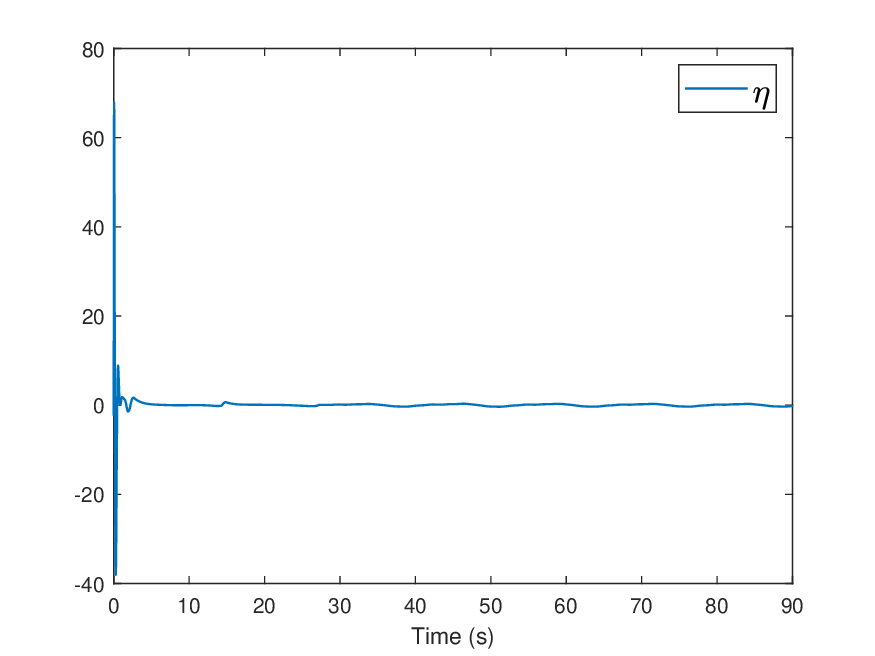}
            \label{subfig:eta}}
    \end{subfigmatrix}
\caption{$u_{\rm drj}$ of the integrating DOBAC and the auxiliary variable $\eta$. Dash: $-\hat{d}$; solid: $u_{\rm drj}$.}
\label{fig:umod_Dynamical}
\end{figure}

\section{Remarks on the role of $k_{\eta}$}
\label{sec:KEtaRemarks}
The parameter $k_{\eta}$ introduced in~\eqref{f:udotStar} may be regarded as a \emph{similarity index} for I-DOBAC relative to D-DOBAC. Supposing $|u_{\rm drj}| < \bar{u}_{\rm drj}$,
\eqref{f:LypunovFunciton_etahat} implies
\begin{align*}
    |\eta| \leq \max{\left\{|\eta(t_0)|e^{-k_{\eta}(t-t_0)},\frac{|d_{\eta}|}{k_{\eta}}\right\}}
\end{align*}
As $k_{\eta}$ increases, the convergence rate of $\eta$ increases and its ultimate bound $\frac{|d_{\eta}|}{k_{\eta}}$ decreases. That is, there exists a decreasing positive function $T_{\eta}^*(k_{\eta})$ such that
$$
    |\eta| \leq \frac{|d_{\eta}|}{k_{\eta}}
    \quad\forall\; t \geq T_{\eta}^*(k_{\eta})
$$
For a sufficiently large $k_{\eta}$, we have
\begin{align*} \label{f:umod_approximateStaic}
    u_{\rm drj} \approx -\hat{d}
    \quad\forall\; t \geq T_{\eta}^*(k_{\eta})
\end{align*}
meaning the I-DOBAC approximates the D-DOBAC after a short transient period. Figure~\ref{subfig:x1Tracking_DynamicalDOBAC_k1000} shows the $x_1$ tracking performance using $k_{\eta} = 1000$; the result is similar to that shown in Figure~\ref{subfig:x1Tracking_StaticDOBAC}. We can see from Figure~\ref{subfig:umod_DynamicalDOBAC_k1000} that $|u_{\rm drj}| < \bar{u}_{\rm drj}$ after the transient response.

Finally note that when $k_{\eta}$ is very small, we have
\begin{align*}
    \dot{\eta} \approx -e_{\dot{\hat{d}}}
\end{align*}
which produces excessive fluctuations in $u_{\rm drj}$ and poor tracking performance; see Figure~\ref{fig:x1Tracking&umod_DynamicalDOBAC_k0d001}. Fortunately, the design resets $u_{\rm drj}$ when it reaches an infeasible value, and a bounded $u_{\rm drj}$ guarantees a bounded $x_1$ by Lemma~\ref{lem:Bounded_x}. Figure~\ref{subfig:umod_DynamicalDOBAC_k0d001} shows $u_{\rm drj}$ remains in the assigned bound.

\begin{figure}[h!]
\centering
    \begin{subfigmatrix}{2}
        \subfigure[Dash: the desired trajectory $x_{1{\rm r}}$; solid: $x_1$.]{
            \includegraphics[width=0.475\textwidth]{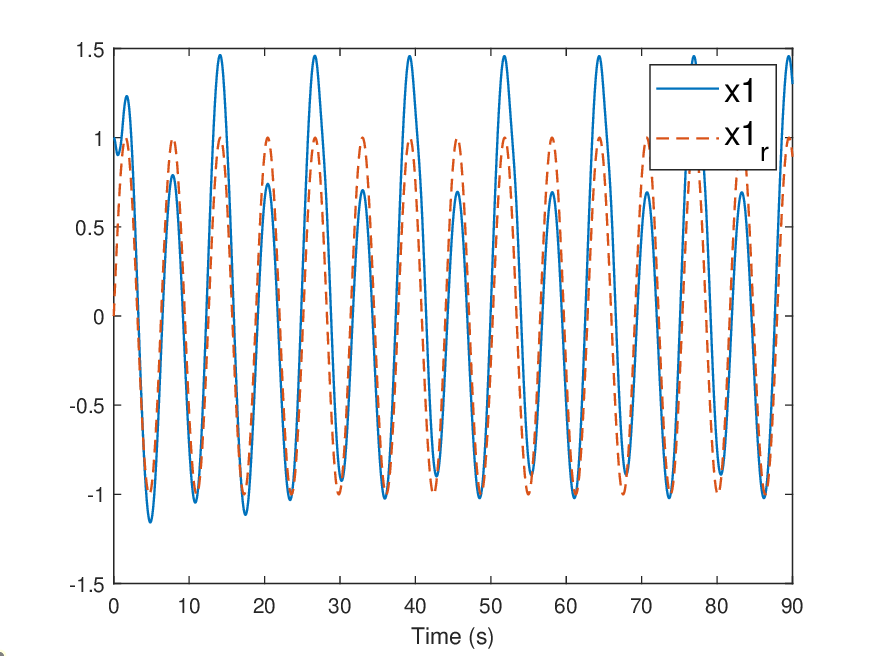}
            \label{subfig:x1Tracking_DynamicalDOBAC_k1000}}
        \subfigure[Dash: $-\hat{d}$; solid: $u_{\rm drj}$.]{
            \includegraphics[width=0.475\textwidth]{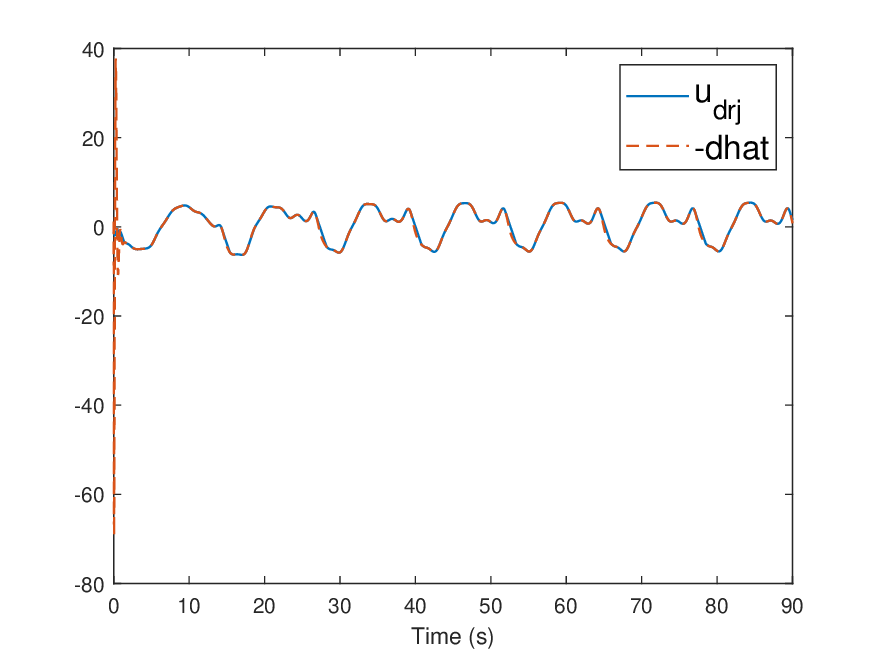}
            \label{subfig:umod_DynamicalDOBAC_k1000}}
    \end{subfigmatrix}
\caption{$x_1$ tracking and $u_{\rm drj}$ with the I-DOBAC using $k_{\eta}=1000$.}
\label{fig:x1Tracking&umod_DynamicalDOBAC_k1000}
\end{figure}
\begin{figure}[h!]
\centering
    \begin{subfigmatrix}{2}
        \subfigure[Dash: the desired trajectory $x_{1{\rm r}}$; solid: $x_1$.]{
            \includegraphics[width=0.475\textwidth]{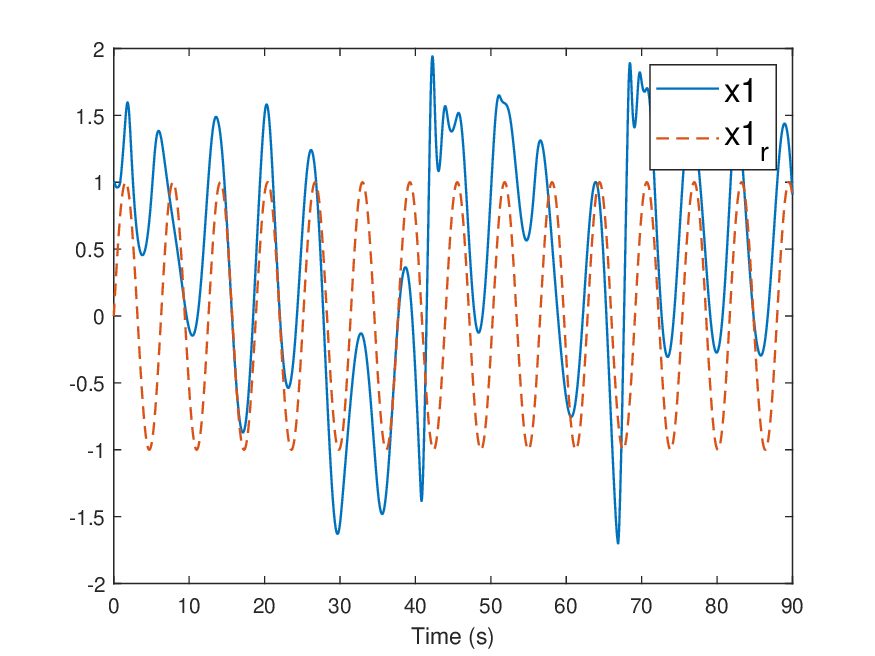}
            \label{subfig:x1Tracking_DynamicalDOBAC_k0d001}}
        \subfigure[Dash: $-\hat{d}$; solid: $u_{\rm drj}$.]{
            \includegraphics[width=0.475\textwidth]{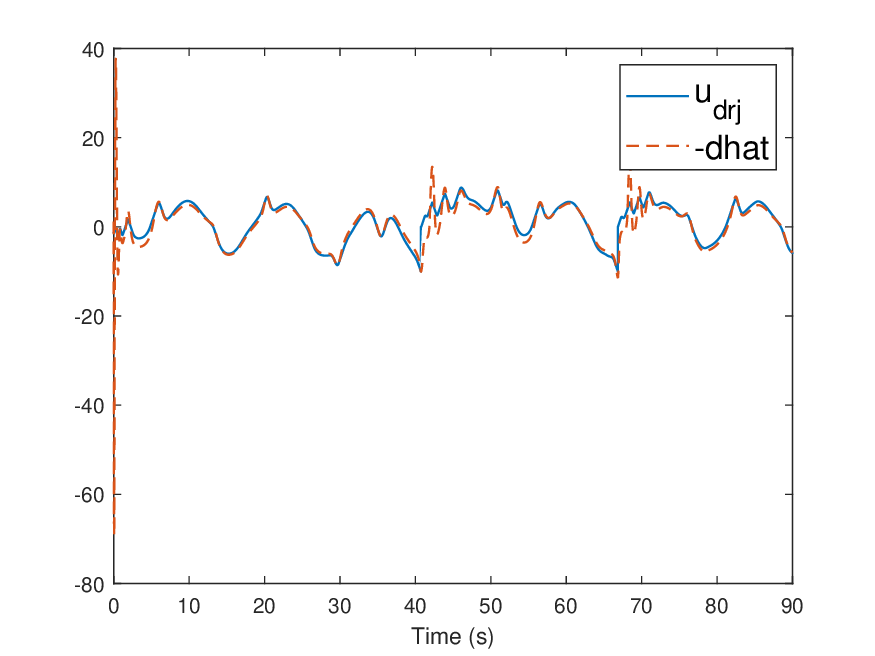}
            \label{subfig:umod_DynamicalDOBAC_k0d001}}
    \end{subfigmatrix}
\caption{$x_1$ tracking and $u_{\rm drj}$ with the I-DOBAC using $k_{\eta}=0.001$.}
\label{fig:x1Tracking&umod_DynamicalDOBAC_k0d001}
\end{figure}

\section{Conclusions}
\label{sec:Conclusion}
An extension to conventional model reference adaptive control is proposed that enhances closed-loop system performance by estimating and actively countering a time-varying disturbance. Any disturbance observer can be incorporated into the proposed scheme and there is no need to redesign the underlying model reference adaptive control law. The disturbance rejection strategy incorporates magnitude- and rate-limited integral action to prevent fast-varying components of the disturbance estimate from affecting the control input. A comparison of conventional model reference adaptive control, ``direct'' disturbance observer-based adaptive control without integral action, and the proposed method shows the integral action improves the accuracy of the disturbance estimate and thus provides better disturbance rejection.

\clearpage
\bibliographystyle{elsarticle-num}
\bibliography{YingchunBib}

\end{document}